\def\year{2019}\relax
\documentclass[letterpaper]{article} 
\usepackage{aaai19}  
\usepackage{times}  
\usepackage{helvet}  
\usepackage{courier}  
\usepackage{url}  
\usepackage{graphicx}  
\usepackage[fleqn]{amsmath}
\usepackage{amsthm}
\usepackage{amssymb}
\title{A Generic Approach to Accelerating Belief Propagation based Incomplete Algorithms for DCOPs via A Branch-and-Bound Technique}
\author{Ziyu Chen,\textsuperscript{\rm 1}Xingqiong Jiang,\textsuperscript{\rm 1}Yanchen Deng\textsuperscript{\rm 2,}\thanks{Corresponding author.}\and Dingding Chen\textsuperscript{\rm 1}\and Zhongshi He\textsuperscript{\rm 1}\\
	College of Computer Science, Chongqing University, Chongqing, China\\
	\textsuperscript{\rm 1} 
	\{chenziyu, jxq, dingding, zshe\}@cqu.edu.cn,
	\textsuperscript{\rm 2}
	dyc941126@126.com
}
\begin{document}
\maketitle
\begin{abstract}
	Belief propagation approaches, such as Max-Sum and its variants, are important methods to solve large-scale Distributed Constraint Optimization Problems (DCOPs). However, for problems with $n$-ary constraints, these algorithms face a huge challenge since their computational complexity scales exponentially with the number of variables a function holds. In this paper, we present a generic and easy-to-use method based on a branch-and-bound technique to solve the issue, called Function Decomposing and State Pruning (FDSP). We theoretically prove that FDSP can provide monotonically non-increasing upper bounds and speed up belief propagation based incomplete DCOP algorithms without an effect on solution quality. Also, our empirically evaluation indicates that FDSP can reduce 97\% of the search space at least and effectively accelerate Max-Sum, compared with the state-of-the-art.  
\end{abstract}

\section{Introduction}
Distributed Constraint Optimization Problems (DCOPs) which require agents to coordinate their decisions to optimize a global objective, are a fundamental framework for modeling multi-agent coordination in multi-agent systems \cite{hirayama1997distributed}. Thus, DCOPs are widely deployed in some real world coordination tasks including meeting scheduling \cite{Enembreck2012Distributed}, sensor networks \cite{Farinelli2014Agent}, power networks \cite{fioretto2017distributed}, etc. 

Algorithms for DCOPs can be classified into two categories: complete and incomplete, according to whether they guarantee to find the optimal solution. Complete algorithms \cite{hirayama1997distributed,modi2005adopt,petcu2005scalable,Yeoh2008BnB,vinyals2009generalizing,gershman2009asynchronous} can get the optimal solutions but incur exponential communication or computation overheads since DCOPs are NP-hard. In contrast, incomplete algorithms \cite{Maheswaran2004Distributed,arshad2004distributed,Zhang2005Distributed,ottens2012duct,nguyen2013distributed,okamoto2016distributed} trade accuracy for computation time and memory so that they can be applied to large-scale problems. As a kind of incomplete algorithms based on belief propagation, Max-Sum \cite{Farinelli2008Decentralised} and its variants \cite{rogers2011bounded,zivan2012max,chen2017iterative} have drawn a lot of attention since they can easily be deployed to any DCOP setting. Moreover, they can explicitly handle $n$-ary constraints and more variables per agent \cite{Cerquides2014A}. In more detail, agents in Max-Sum propagate and accumulate beliefs through the whole factor graph. And each agent only holds its belief about the utility for each possible assignment and continuously updates its belief based on the messages received from its neighbors. 

In spite of many advantages of belief propagation approaches, they suffer from a huge challenge in scalability. Specifically, they perform maximization operations repeatedly to locally accumulate beliefs for the involved variables, given the local utility function and a set of incoming messages. The computation complexity of this step grows exponentially as the number of constraint arities. In other words, when a constraint function holds $n$ variables and the domain size of each variable is $d$, Max-Sum needs to perform $d^n$ maximization operations to yield the best assignment for each variable.

To address the issue, two kinds of methods were proposed to improve the scalability of belief propagation approaches. The first kind is the algorithms based on a branch-and-bound technique including BnB-MS \cite{Stranders2009Decentralised} and BnB-FMS \cite{Macarthur2011A} which both consist of a preprocessing phase and a pruning phase. In the preprocessing phase, the two algorithms use localized message-passing to simplify DCOPs. Specifically, BnB-MS reduces the number of moves that each agent needs to consider in coordinating mobile sensors while BnB-FMS removes tasks that an agent should never perform in dynamic task allocations. In the pruning phase, both algorithms reduce the search space using a branch-and-bound technique to speed up maximization operations. Unfortunately, these algorithms require to exchange a lot of messages in their preprocessing phases. Moreover, the bounds in these algorithms are computed by either brute force or domain-specific knowledge, which limits their applicability.

The second kind of approaches is sorting based, such as G-FBP\cite{Kim2013Improved} and GDP\cite{khan2018generic}, which is applicable to all DCOP settings. G-FBP uses partially sorted lists to adapt FBP. Specifically, it selects and sorts the top $cd^{\frac{n-1}{2}}$ values of the search space, presuming that the maximum value can be found from the selected range. Here, $c$ is a constant. However, G-FBP will incur additional computation once the maximum value cannot be found within the selected range. Different from G-FBP, the main idea of GDP is to explore only the rows that can cover the differences between the sum of the maximal utility of each message and the message utility corresponding to the assignment that produces the largest local utility. Thus, GDP needs to sort the local utilities of all function-nodes independently by each assignment of each variable in the preprocessing phase and $\mathcal{V}_i$ is the sorted result of each assignment $i$. Then, GDP returns a pruned range $[p,q]$ or $[p,q)$ according to whether $q==p-t$, where $p=\max(\mathcal{V}_i)$, $q=\max_c\{c\in \mathcal{V}_i: c\le (p-t)\}$ and $t=m-b$. Here, $m$ is the summation of the maximum value for each received message from other variable-nodes and $b$ is the summation for the corresponding values of $p$ from the incoming messages of a function-node. However, GDP needs additional time to perform sorting operations in the preprocessing phase. More importantly, GDP is an one-shot pruning procedure that cannot use the learned experience from the assignment combinations explored to dynamically prune the search space.   

Given the background, we devote to develop a generic and fast method for belief propagation based on a branch-and-bound technique, called Function Decomposing and State Pruning (FDSP). In more detail, we propose a domain-independent approach based on dynamic programming to effectively evaluate the upper bound of a given partial assignment, which overcomes the aforementioned drawbacks of BnB-MS and BnB-FMS. We further enforce the upper bounds by exploiting the fact that the assignment of the target variable is given. Finally, we prune the search space whenever the upper bound of a partial assignment is no greater than the best lower bound explored so far. The experimental results show the effectiveness of FDSP which can prune at least 97\% of the search space when solving complex problems. 

\section{Background}
\subsection{Distributed Constraint Optimization Problems}
A distributed constraint optimization problem can be represented by a tuple $\left\langle A,X,D,F \right\rangle$ such that:
\begin{itemize}
	\item $A=\{a_1,a_2,\dots,a_h\}$ is a set of agents.
	\item $X=\{x_1,x_2,\dots,x_q\}$ is a set of variables.
	\item $D=\{D_1,D_2,\dots,D_q\}$ is a set of finite and discrete domains, variable $x_i$ taking an assignment value in $D_i$.
	\item $F=\{F_1,F_2,\dots,F_r\}$ is a set of constraints, where each constraint $F_k: \mathbf{x_k}\rightarrow \mathbb{R}^{+}$ denotes how much utility is assigned to each possible combination of assignments of the involved variables $\mathbf{x_k}\subseteq X$. 
\end{itemize}
Thus, a constraint function $F_k(\mathbf{x_k})$ denotes the utility for each possible assignment combination of the variables in $\mathbf{x_k}$, $n=|\mathbf{x_k}|$ represents the arity of $F_k$, and $d=|D_i|$ denotes the domain size of variable $x_i$. Note that the variables in $\mathbf{x_k}$ are ordered according to their own indices, where a variable $\mathbf{x_{k,i}}$ is ordered before a variable $\mathbf{x_{k,j}}$ if $i<j$.     

Given this, the goal for a DCOP is to find the joint variable assignment $X^*$ such that a given global objective function is maximal. Generally, the objective function is described as the sum over $F$: 
$$
X^*=\mathop{\arg\max}_{X}\sum_{F_k(\mathbf{x_k})\in F,\mathbf{x_k}\subseteq X}F_k(\mathbf{x_k}) 
$$
\begin{figure}
	\centering
	\includegraphics[scale=0.7]{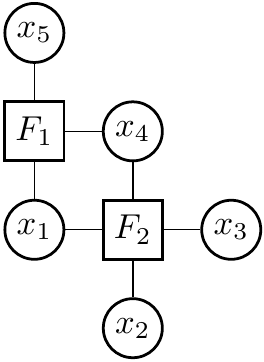}
	\caption{A DCOP instance using a factor graph}
\end{figure}
\subsection{Max-Sum Algorithm}
As a belief propagation approach, Max-Sum is a message-passing inference-based algorithm operating on factor graphs which comprise variable-nodes and function-nodes. Function-nodes which represent the constraints in a DCOP are connected to variable-nodes they depend on, while variable-nodes which represent the variables in a DCOP are connected to function-nodes they are involved in. As shown in Fig. 1, $F_1$ and $F_2$ are two function-nodes, $x_1$, $x_2$, $x_3$, $x_4$ and $x_5$ are variable-nodes, where $x_1$, $x_2$, $x_3$ and $x_4$ are connected to $F_2$, and $x_1$, $x_4$ and $x_5$ are connected to $F_1$. Here, $F_1$ (i.e., $F_1(\mathbf{x_1})$, where $\mathbf{x_1}=\{x_1,x_4,x_5\}$) is a $3$-ary constraint and $F_2$ (i.e., $F_2(\mathbf{x_2})$, where $\mathbf{x_2}=\{x_1,x_2,x_3,x_4\}$) is a $4$-ary one.  

In Max-Sum, beliefs are propagated and accumulated through the whole factor graph via the messages exchanged between variable-nodes and function-nodes. The message from a variable-node $x_i$ to a function-node $F_k\left(\mathbf{x_k}\right)$, called the query message. It is defined by
\begin{equation}
Q_{x_i\rightarrow F_k}\left(x_i\right)=\alpha_{ik}+\sum_{F_j\in N_i\backslash F_k}R_{F_j\rightarrow x_i}(x_i)
\end{equation}
where $\alpha_{ik}$ is a scalar set such that $\sum_{x_i}Q_{x_i\rightarrow F_k}\left(x_i\right) = 0$, $x_i\in \mathbf{x_k}$ and $N_i\backslash F_k$ is a set of neighbors of $x_i$ except the target function-node $F_k$. The response message sent from a function-node $F_k\left(\mathbf{x_k}\right)$ to a variable-node $x_i$ is given by
\begin{equation}
\resizebox{.9\hsize}{!}{$
	R_{F_k\rightarrow x_i}(x_i)=\max\limits_{\mathbf{x_k}\backslash x_i}\left(F_k\left(\mathbf{x_k}\right)+\sum\limits_{x_j\in \mathbf{x_k}\backslash x_i}Q_{x_j\rightarrow F_k}\left(x_j\right)\right) \label{con:e2}$}
\end{equation}

When a variable-node $x_i$ makes its decision, it first accumulates the belief for each possible assignment from all messages it receives. Then, it selects a value to maximize the total utilities. The procedure can be formalized by:
\begin{equation}
x_i^*=\mathop{\arg\max}_{x_i}\sum_{F_k\in N_i}R_{F_k\rightarrow x_i}(x_i)
\end{equation} 
\begin{figure}[]
	\begin{minipage}[t]{0.55\linewidth}
		\includegraphics[scale=0.7]{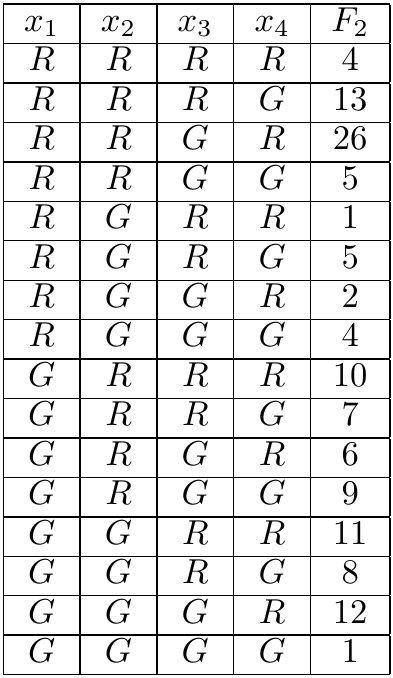}
		\\
		\centering
		(a)
	\end{minipage}%
	\begin{minipage}[t]{0.45\linewidth}
		\includegraphics[scale=0.7]{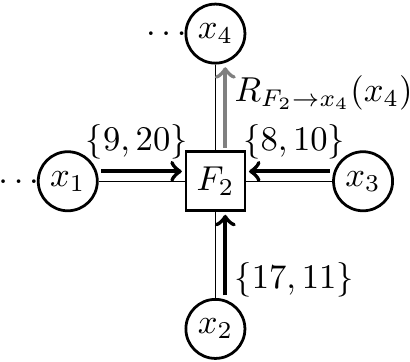}
		\\
		\centering
		(b)
	\end{minipage}%
	\caption{The utility matrix and message exchange for $F_2$} 
\end{figure}
In this paper, we use the term "target variable" to denote the destination variable-node of the outgoing message being computed, and $\mathcal{M}_{\mathbf{x_{k,i}}}$ represents the message from variable $\mathbf{x_{k,i}}$.   
\section{Proposed Method}
\subsection{Motivation}
The maximization operation in Eq. (\ref{con:e2}) is the most computationally expensive operation in Max-Sum. For a $n$-ary DCOP, the complexity of computing the response message is $O(d^n)$. Take Fig. 2 as an example. Assume that each variable takes a value in $\{R,G\}$ and the function-node $F_2$ has received the messages $\mathcal{M}_{x_1}=\{9,20\}$, $\mathcal{M}_{x_2}=\{17,11\}$ and $\mathcal{M}_{x_3}=\{8,10\}$ from $x_1$, $x_2$ and $x_3$, respectively. Then, $F_2$ requires $d^n=2^4=16$ operations to generate the message $R_{F_2\rightarrow x_4}(x_4)$ since its domain size $d=2$ and arity $n=4$. Obviously, the complexity of this step grows exponentially as $d$ and $n$ scale up. Therefore, this is a huge challenge for scalability of belief propagation algorithms.      

As mentioned earlier, some efforts have been made to optimize this maximization operation. Nevertheless, the improved algorithms based on a branch-and-bound technique including BnB-MS and BnB-FMS require a number of messages to be passed in the preprocessing phase. And, these algorithms were proposed for the exact application, which makes it difficult to directly solve general DCOPs. Besides, the algorithms based on sorting, such as G-FBP and GDP, suffer from some drawbacks although they are generic. Specifically, G-FBP cannot guarantee that the maximum value can be found in the selected range, which can lead to a complete traverse to all the possible combinations in the worst case, while GDP requires sorting for each value in the domain of each variable in the preprocessing phase, which makes its use prohibitively expensive. Additionally, GDP cannot use the learned knowledge from the combinations explored to dynamically prune the search space. In other words, GDP is actually an one-shot pruning method. Taking Fig. 2 for example, according to GDP, the local utilities of $F_2$  are sorted independently by each value of the domain. When computing the pruned range of value $G$, there are $\mathcal{V}_G=\{13,9,8,7,5,5,4,1\}$, $p=\max(\mathcal{V}_G)=13$ and the base case $t=(20+17+10)-(9+17+8)=13$. Hence, $p-t=0$. Accordingly, GDP returns a fixed pruned range $[13,1]$ for value $G$. Obviously, the pruned range contains the entire search space of value $G$, and cannot be reduced in the subsequent search process.              

Under such circumstances, we propose a generic, fast and easy-to-use approach based on a branch-and-bound technique, called FDSP that can use the learned experience from the combinations explored to dynamically prune the search space.      
\subsection{FDSP} 
FDSP generally consists of two components: estimation to provide upper bounds and pruning to reduce the search space. To provide the optimal upper bound for a partial assignment, the estimation must return the upper bounds for both the local function and the incoming messages. FDSP computes the function estimations in the preprocessing phase, called Function Decomposing (FD), while the message estimations are (re)constructed once the messages have changed. Pruning is implemented by a procedure called State Pruning (SP) which is based on a branch and bound technique. That is, the algorithm does not expand any partial assignment whose upper bound is less than the known lower bound. FDSP can be easily applied to any belief propagation based incomplete algorithms to deal with DCOPs with $n$-ary constraints. 
\subsubsection{Function Decomposing} serves in a preprocessing phase to compute the function estimation for each variable of a function-node $F_k(\mathbf{x_k})$. Given a partial assignment $PA|_{\mathbf{x_{k,1}}}^{\mathbf{x_{k,i}}}$ to variables $\{\mathbf{x_{k,w}}\in \mathbf{x_k}|1\le w\le i\}$, the upper bound of the local function is maximization of $F_k(\mathbf{x_k})$ over the remaining unassigned variables. That is
\begin{equation}
\resizebox{.9\hsize}{!}{$FunEst_{\mathbf{x_{k,i}}}(PA|_{\mathbf{x_{k,1}}}^{\mathbf{x_{k,i}}})=\max\limits_{z=\{\mathbf{x_{k,j}}|j>i\}}F_k(PA|_{\mathbf{x}_{\mathbf{k,1}}}^{\mathbf{x}_{\mathbf{k,i}}},z)$} \label{con:e4}
\end{equation}
Here, $FunEst_{\mathbf{x_{k,i}}}(\boldsymbol{\cdot})$ is the \textit{uninformed} function estimation for the $i$-th variable $\mathbf{x_{k,i}}$ in $\mathbf{x_k}$, which provides optimistic upper bounds on the utilities of the subsequent search spaces of $PA|_{\mathbf{x_{k,1}}}^{\mathbf{x_{k,i}}}$. Macarthur et al. \shortcite{Macarthur2011A} tried to compute the estimation by using brute force, which incurs exponential operations for each partial assignment. They also suggested that the domain characteristics can be used to compute the estimation, which has a limited generalization and cannot guarantee the tightness. In contrast, our proposed FD is an one-shot preprocessing procedure that uses dynamic programming to compute the estimation, which can significantly reduce the computation efforts. Specifically, the estimations are computed recursively according to Eq. (\ref{con:e5}). 
\begin{equation}
\resizebox{.9\hsize}{!}{$FunEst_{\mathbf{x}_{\mathbf{k,i}}}=\begin{cases}
	F_k\left( \mathbf{x}_{\mathbf{k}} \right) \ \ \ \ \ \ \ \ \ \ \ \ \ \ \ \ \ \ \ \ \ \ \ \ \ \ i=n\\
	\underset{\mathbf{x}_{\mathbf{k,i}+\mathbf{1}}}{\max}\ FunEst_{\mathbf{x}_{\mathbf{k,i}+1}}\ \ \ \ \ \ otherwise\\
	\end{cases}$}\label{con:e5}
\end{equation}
That is, the estimation for a variable is maximization of the one for the next variable. Particularly, the estimation for the last variable is the function itself. Note that, compared to the exponential overhead for each partial assignment in BnB-MS, our proposed FD only requires $O(d^{i+1})$ operations to compute the function estimation for each variable $\mathbf{x_{k,i}}$ in the preprocessing phase.

In fact, the uninformed function estimation $FunEst_{\mathbf{x_{k,i}}}$ for $\mathbf{x_{k,i}}$ could provide a tighter upper bound if we know the assignment of a variable $\mathbf{x_{k,j}}$ such that $j>i$. In this way, we can compute a tight upper bound even if there are many unassigned variables between the last assigned variable and the target variable in a partial assignment. By considering all the possible assignments of each variable ordered after $\mathbf{x_{k,i}}$, we further reinforce the upper bound and propose the \textit{informed} function estimations. Eq. (\ref{com:e6}) gives the formal definition to the informed function estimation for $\mathbf{x_{k,i}}$ in terms of $\mathbf{x_{k,j}}$ where $j>i$.  
\begin{equation}
\resizebox{.9\hsize}{!}{$FunEst_{\mathbf{x}_{\mathbf{k,i}}}^{\mathbf{x}_{\mathbf{k,j}}=v_{k,j}}=\begin{cases}
	FunEst_{\mathbf{x}_{\mathbf{k,j}}}\left( v_{k,j} \right) \ \ \ \ \ \ \ \ \ \ i=j-1\\
	\underset{\mathbf{x}_{\mathbf{k,i}+\mathbf{1}}}{\max}\ FunEst_{\mathbf{x}_{\mathbf{k,i}+1}}^{\mathbf{x}_{\mathbf{k,j}}=v_{k,j}}\ \ otherwise\\
	\end{cases}$}\label{com:e6}
\end{equation}
Similar to the uninformed function estimations, the informed function estimations are computed in a recursive fashion by maximizing the estimation for the next variable. And the estimation for the last variable before $\mathbf{x_{k,j}}$ is the corresponding uninformed function estimation with a given assignment $\mathbf{x_{k,j}}=v_{k,j}$.

Fig. 3 gives the sketch of FD. The procedure begins with computing the uninformed function estimation for each variable in $F_k$ according to Eq. (\ref{con:e5}) from the last one to the first one (line 1-3). Then, for every possible assignment of each variable, we compute the informed function estimation for each variable whose index is smaller than the current variable according to Eq. (\ref{com:e6}) (line 4-6).      
\begin{figure}
	\centering
	\includegraphics[scale=0.8]{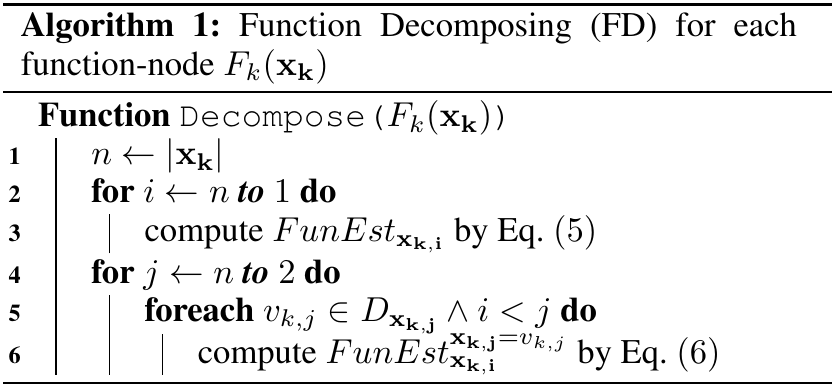}
	\caption{The sketch of function decomposing}
\end{figure}            

Taking Fig. 2 for example, we can compute the uninformed function estimations for variable $x_1$, $x_2$, $x_3$ and $x_4$ as follows:
\begin{tabbing}
	$FunEst_{x_4}=F_2$, $FunEst_{x_3}=\max_{x_4}FunEst_{x_4}$\\
	$FunEst_{x_2}=\max_{x_3}FunEst_{x_3}$\\
	$FunEst_{x_1}=\max_{x_2}FunEst_{x_2}$	
\end{tabbing}
These estimations can provide the upper bounds for the partial assignments with respect to their variables. Besides, the informed function estimations in terms of $x_4=R$ are computed as follows: 
\begin{tabbing}
	$FunEst_{x_3}^{x_4=R}=FunEst_{x_4}(x_4=R)$\\
	$FunEst_{x_2}^{x_4=R}=\max_{x_3}FunEst_{x_3}^{x_4=R}$\\
	$FunEst_{x_1}^{x_4=R}=\max_{x_2}FunEst_{x_2}^{x_4=R}$	
\end{tabbing}

\subsubsection{State Pruning} is geared towards speeding up the computation of the messages from function-nodes to variable-nodes by branch and bound. That is, when the upper bound of a partial assignment is no greater than the lower bound, the search space corresponding to the partial assignment will be discarded. Fig. 4 gives the pseudo code of SP.
\begin{figure}
	\centering
	\includegraphics[scale=0.8]{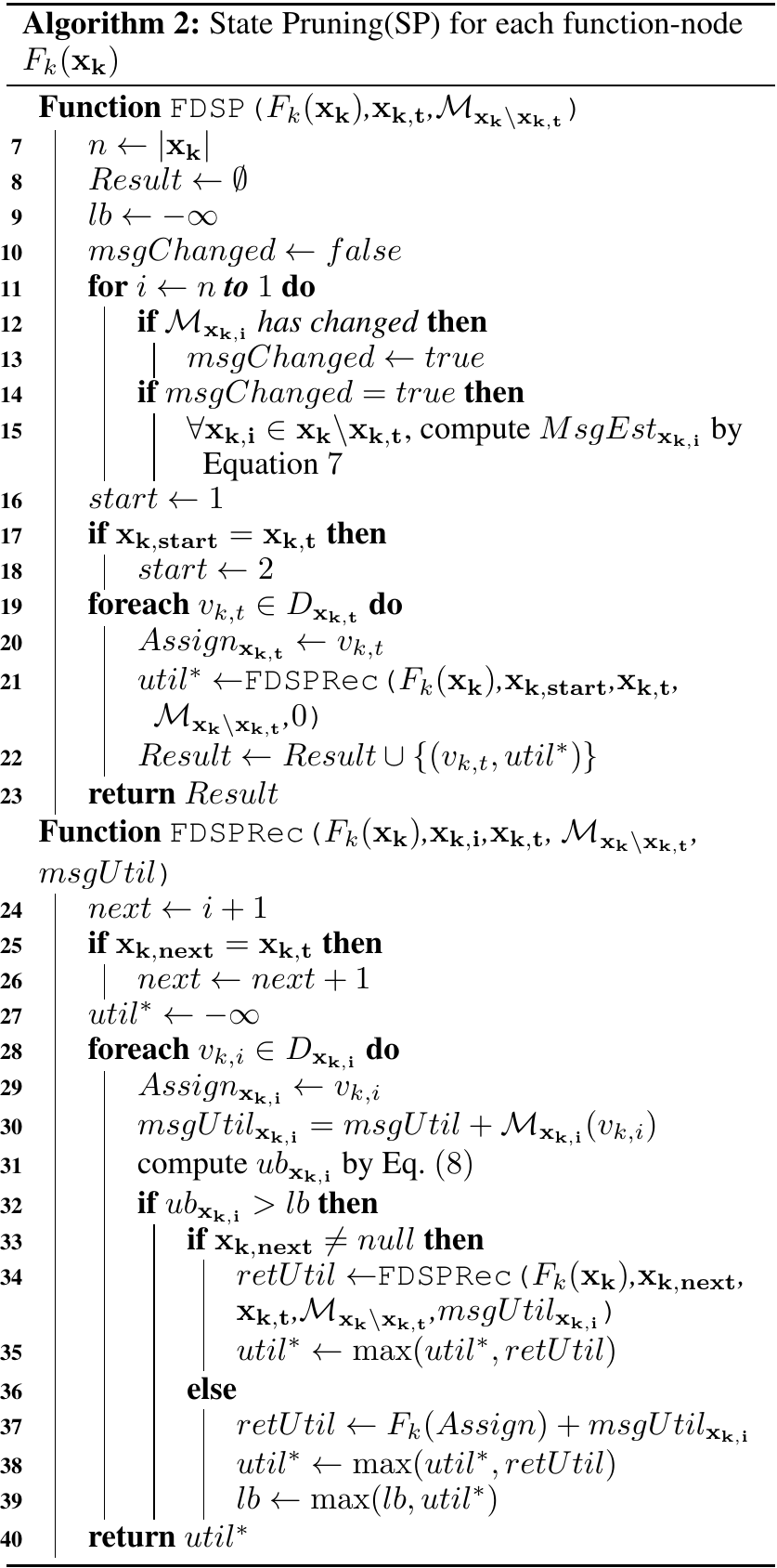}
	\caption{The sketch of state pruning}
\end{figure}

The algorithm begins with calculating the message estimation for each variable $\mathbf{x_{k,i}}\in \mathbf{x_{k}}\backslash \mathbf{x_{k,t}}$, which gives the maximal message utility with respect to all non-target variables after it given these variables unassigned, according to Eq. (\ref{con:e6}) (line 7-15).   
\begin{equation}
MsgEst_{\mathbf{x_{k,i}}}=\sum\limits_{j>i\land j\ne t}\max(\mathcal{M}_{\mathbf{x_{k,j}}})\label{con:e6}
\end{equation}
Here, $\mathbf{x_{k,t}}$ is the target variable and $MsgEst_{\mathbf{x_{k,i}}}$ denotes the upper bounds on the received messages from the variables after $\mathbf{x_{k,i}}$ except $\mathbf{x_{k,t}}$. In order to reduce the unnecessary computation, $F_k$ recomputes the message estimations for each variable ordered before $\mathbf{x_{k,i}}$ only when the message from $\mathbf{x_{k,i}}$ changes. Besides, instead of directly computing message estimations according to Eq. (\ref{con:e6}), $F_k$ further reduces the computation efforts by recursively backing up the maximal message utilities from the last non-target variable to the first one. That is, the message estimation of a variable is the sum of the maximal message utility and the message estimation of the non-target variable next to it. Consider the function-node $F_2$ in Fig. 2(b). When we are computing the message for $x_4$, the message estimations for $x_1$, $x_2$ and $x_3$ are computed as follows:
\begin{tabbing}
	$MsgEst_{x_3}= 0$ \\
	$MsgEst_{x_2}= MsgEst_{x_3}+\max\mathcal{M}_{x_3} = 10 $\\
	$MsgEst_{x_1}= MsgEst_{x_2}+\max\mathcal{M}_{x_2}= 10+17=27$ 
\end{tabbing}    

Then, $F_k$ computes the maximum utility $util^*$ of each assignment of the target variable $\mathbf{x_{k,t}}$ in $D_{\mathbf{x_{k,t}}}$ (line 16-22). Specifically, $F_k$ assigns assignment $v_{k,t}$ to $\mathbf{x_{k,t}}$ according to the order of values in $D_{\mathbf{x_{k,t}}}$ (line 19). Thus, the current partial joint state $Assign=\{\emptyset,\ldots,v_{k,t},\ldots,\emptyset\}$, where $\emptyset$ represents an unassigned variable (line 20). After that, \textbf{FDSPRec} is called for the variable $\mathbf{x_{k,start}}$ which is the first unassigned variable to recursively expand the search space (line 21). Finally, $F_k$ stores $util^*$ to $Result$ when $util^*$ for the current assignment $v_{k,t}$ is returned (line 22). The procedure (line 20 - 22) repeats until all the assignments of $\mathbf{x_{k,t}}$ have been visited.

In \textbf{FDSPRec}, $F_k$ first finds $\mathbf{x_{k,next}}$ that is the unassigned variable next to $\mathbf{x_{k,i}}$ (line 24-26). Note that $\mathbf{x_{k,t}}$ is an assigned variable. Then, $F_k$ decides to expand the search space or update the maximum utility and the lower bound (line 28-39). In more detail, $F_k$ expands the search space by appending the assignment $v_{k,i}$ of $\mathbf{x_{k,i}}$ to the partial joint state (line 29). And then, $F_k$ computes the utilities contributed (i.e., $msgUtil_{\mathbf{x_{k,i}}}$) by the incoming messages with respect to the current $Assign$ by summing the accumulated $msgUtil$ with the entry in terms of $\mathcal{M}_{\mathbf{x_{k,i}}}$ and assignment $v_{k,i}$ (line 30). Then, $F_k$ computes the current upper bound $ub_{\mathbf{x_{k,i}}}$ according to Eq. (\ref{com:e7}) (line 31).      
\begin{equation}
\resizebox{1.04\hsize}{!}{$ub_{\mathbf{x}_{\mathbf{k,i}}}\!=\! \begin{cases}
	msgUtil_{\mathbf{x}_{\mathbf{k,i}}}+MsgEst_{\mathbf{x}_{\mathbf{k,i}}}+FunEst_{\mathbf{x}_{\mathbf{k,i}}}\left( Assign|_{\mathbf{x}_{\mathbf{k,1}}}^{\mathbf{x}_{\mathbf{k,i}}} \right) \ \ \ \ \ \ \ \ \  i>t\\
	msgUtil_{\mathbf{x}_{\mathbf{k,i}}}+MsgEst_{\mathbf{x}_{\mathbf{k,i}}}+FunEst_{\mathbf{x}_{\mathbf{k,i}}}^{\mathbf{x}_{\mathbf{k,t}}=v_{k,t}}\left( Assign|_{\mathbf{x}_{\mathbf{k,1}}}^{\mathbf{x}_{\mathbf{k,i}}} \right) \  i<t\\
	\end{cases}$}  \label{com:e7}
\end{equation}            
Specifically, if $i>t$(i.e., $\mathbf{x_{k,i}}$ is after $\mathbf{x_{k,t}}$), which means the assignments of all variables before $\mathbf{x_{k,t}}$ have been given, the current upper bound of the local function is provided by the uninformed function estimation of variable $\mathbf{x_{k,i}}$. Otherwise, the upper bound is computed by the informed function estimation. In other words, the informed function estimation is used to compute a tight upper bound whenever it is applicable.  

Next, $F_k$ decides whether to expand the search space according to the lower bound (line 32-39). If the upper bound $ub_{\mathbf{x_{k,i}}}$ is greater than the current lower bound $lb$ and $\mathbf{x_{k,i}}$ is not the last variable, the algorithm proceeds by calling the recursive function \textbf{FDSPRec} to expand the search space (line 33-35). Otherwise, the search space corresponding to $Assign$ can be discarded. If $\mathbf{x_{k,i}}$ is the last non-target variable, i.e., the search space has been fully expanded, $F_k$ computes the current utility $retUtil$ of the complete assignment by adding the local utility $F_k(Assign)$ and $msgUtil_{\mathbf{x_{k,i}}}$. Then, $F_k$ updates the maximum utility $util^*$ and the lower bound $lb$ (line 36-39). Finally, when all the assignments of $\mathbf{x_{k,i}}$ have been visited, the algorithm returns $util^*$ (line 40).  
\begin{figure}
	\centering
	\includegraphics[scale=0.67]{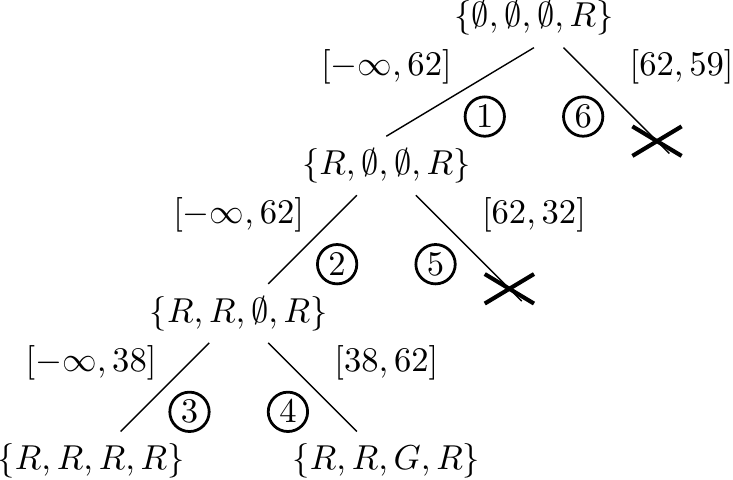}
	\caption{Calculating $R_{F_2\rightarrow x_4}(x_4=R)$ using FDSP}
\end{figure}

Fig. 5 shows an example for calculating the message from function $F_2$ to variable $x_4$ when $x_4=R$ (i.e., $R_{F_2\rightarrow x_4}(x_4=R)$) in Fig. 2, where the numbers with circles represent the trace of SP. Since $x_4$ is fixed to assignment $R$, $F_2$ needs to compute the maximum utility $util^*$ by extending the partial assignment $Assign=\{\emptyset,\emptyset,\emptyset,R\}$. Firstly, $F_2$ visits the first assignment $R$ of $x_1$ and computes $ub_{x_1}=9+27+26=62$ by Eq. (\ref{com:e7}). Then, it expands $Assign=\{R,\emptyset,\emptyset,R\}$ by visiting the first assignment $R$ of $x_2$ since $ub_{x_1}>lb(=-\infty)$. Similarly, $F_2$ expands $Assign=\{R,R,\emptyset,R\}$. At this point, since $ub_{x_3}>lb$ and $Assign$ is fully assigned, $F_2$ computes the utility $retUtil$ of $Assign=\{R,R,R,R\}$: $retUtil=F_2(Assign)+msgUtil_{x_3}=4+(9+17+8)=38$, and updates $util^*=38$ and $lb=38$. Next, $F_2$ visits the next assignment $G$ of $x_3$. Similarly, $F_2$ computes the current upper bound $ub_{x_3}=62$ and the current utility $retUtil=62$ corresponding to $Assign=\{R,R,G,R\}$, and updates $lb$ and $util^*$. After that, $F_2$ visits the second assignment $G$ of $x_2$ since all assignments of $x_3$ have been exhausted. And, $F_2$ computes $ub_{x_2}=32$ which is less than $lb$, so $Assign=\{R,G,\emptyset,R\}$ is discarded. Similarly, $Assign=\{G,\emptyset,\emptyset,R\}$ is also discarded. Finally, $F_2$ finds the maximum utility $util^*=62$, i.e., $R_{F_2\rightarrow x_4}(x_4=R)=62$.                          

As seen from the example, FDSP can prune at least 75\% of the search space during computing the message from the function-node $F_2$ to variable-node $x_4$, where $d=2$ and $n=4$.                        
\subsection{Theoretical Analysis}
In this section, we will theoretically prove that FDSP can speed up belief propagation based incomplete algorithms without an effect on solution quality, i.e., FDSP can provide monotonically non-increasing upper bounds and never prunes the optimal assignment with the maximum utility $util^*$.
\newtheorem{theorem}{Theorem} 
\newtheorem{lemma}{Lemma}
\begin{lemma}
	For a function-node $F_{k}(\textbf{x}_\textbf{k})$ and a given partial assignment $PA$ with $(\textbf{x}_{\textbf{k,t}}=v_{k,t})$ in which $(\textbf{x}_{\textbf{k,i}}=v_{k,i})$ is the last non-target entry, the upper bound of any direct subsequent partial assignment $PA^\prime=PA\cup (\textbf{x}_{\textbf{k,j}}=v_{k,j})$ is at least as low as the one of $PA$, where $\textbf{x}_{\textbf{k,j}}$ is the variable next to $\textbf{x}_{\textbf{k,i}}$ such that $j\ne t$.
\end{lemma}
\begin{proof}
	Recall that the upper bound of a given partial assignment is computed according to either the uninformed function estimation or the informed function estimation, depending on the index of the target variable. Thus, three cases need to be discussed: 1) all the upper bounds are computed according to the uninformed function estimations; 2) the upper bound of $PA$ is computed according to the uninformed function estimation, while the one of $PA^\prime$ is computed according to the informed function estimation; 3) all the upper bounds are computed according to the informed function estimations. Here, we only give the prove for case 2) (i.e., $i+1=t$, $t+1=j$) due to the limited space. Similar analysis can be applied to case 1) and 3). 
	$$
	\scriptsize
	\begin{aligned}
	ub_{\textbf{x}_\textbf{k,i}}(PA)&=FunEst_{\textbf{x}_\textbf{k,i}}^{\textbf{x}_\textbf{k,t}=v_{k,t}}(PA)+msgUtil_{\textbf{x}_\textbf{k,i}}+MsgEst_{\textbf{x}_\textbf{k,i}}\\
	&=FunEst_{\textbf{x}_\textbf{k,t}}(PA)+\sum_{l\le i}\mathcal{M}_{\textbf{x}_\textbf{k,l}}(v_{k,l})+\sum_{l>i\land l\ne t}\max(\mathcal{M}_{\textbf{x}_\textbf{k,l}})\\
	&\ge FunEst_{\textbf{x}_\textbf{k,j}}(PA^\prime)+\sum_{l\le i}\mathcal{M}_{\textbf{x}_\textbf{k,l}}(v_{k,l})+\sum_{l>i\land l\ne t}\max(\mathcal{M}_{\textbf{x}_\textbf{k,l}})\\
	&\ge FunEst_{\textbf{x}_\textbf{k,j}}(PA^\prime)+\sum_{l\le j\land l\ne t}\mathcal{M}_{\textbf{x}_\textbf{k,l}}(v_{k,l})+\sum_{l>j}\max(\mathcal{M}_{\textbf{x}_\textbf{k,l}})\\
	&=ub_{\textbf{x}_\textbf{k,j}}(PA^\prime)
	\end{aligned}
	$$
	Here, the second step holds since $i=t-1$. Thus, according to Eq. (\ref{com:e6}) we have $FunEst_{\textbf{x}_\textbf{k,i}}^{\textbf{x}_\textbf{k,t}=v_{k,t}}(PA)=FunEst_{\textbf{x}_\textbf{k,t}}(PA)$. Besides, the third step and the fourth step hold since $FunEst_{\textbf{x}_\textbf{k,t}}(PA)=\max_{\textbf{x}_\textbf{k,j}} FunEst_{\textbf{x}_\textbf{k,j}}(PA,\mathbf{x_{k,j}})\ge FunEst_{\textbf{x}_\textbf{k,j}}(PA^\prime)$(Eq. (\ref{con:e5})) and $\max(\mathcal{M}_{\textbf{x}_\textbf{k,j}})\ge \mathcal{M}_{\textbf{x}_\textbf{k,j}}(v_{k,j})$, respectively.  
	
	Thus the lemma is proved.
\end{proof}

\begin{theorem}
	FDSP does not affect the optimality of Eq. (2).
\end{theorem}
\begin{proof}
	Prove by contradiction. For a function-node $F_{k}(\textbf{x}_\textbf{k})$, assume that the optimal assignment of Eq. (2) is $Assign^*$, and the corresponding utility value is $val(Assign^*)$. Assume that FDSP has missed that assignment. Thus, there must exist a partial assignment $PA\subset Assign^*$ such that $ub(PA)<lb \le val(Assign^*)$. According to Lemma 1, the upper bound is monotonically non-increasing, i.e., $ub(PA)\ge ub(Assign^*)$. Note that
	$$
	\scriptsize
	\begin{aligned}
	ub_{\textbf{x}_\textbf{k,n}}(Assign^*)&=FunEst_{\textbf{x}_\textbf{k,n}}(Assign^*)+msgUtil_{\textbf{x}_\textbf{k,n}}+MsgEst_{\textbf{x}_\textbf{k,n}}\\
	&=F_k(Assign^*)+\sum_{l\le n\land l\ne t}\mathcal{M}_{\textbf{x}_{\textbf{k,l}}(v_{k,l})}\\
	&=val(Assign^*)
	\end{aligned}
	$$
	Here, $n=|\mathbf{x_k}|$. Thus, we have $ub(PA)\ge ub(Assign^*)=val(Assign^*)$, which is contradict to the assumption. Therefore, the upper bound of a partial assignment cannot be less than the value of any subsequent full assignment and the optimality is hereby guaranteed.
\end{proof}
\subsection{Complexity Analysis}
Each variable $\mathbf{x_{k,i}}$ needs to compute and stores an uninformed function estimation and $d(n-i)$ informed function estimations in the preprocessing phase. Thus, the time and space of each variable require $O([1+d(n-i)]d^{i+1})$ and $O([1+d(n-i)]d^{i})$, respectively, where the value of $i$ becomes smaller as FD performs. Thus, FDSP in the preprocessing phase needs a small overhead.  

Besides, since each function-node needs to explore the search space with respect to the target variable, the time complexity in the worst case is $O(d^n)$. However, with SP, only the small search space needs to be explored. Therefore, the overall overhead is small. For this point, our empirical evaluation also verifies that FDSP only requires little time to run.       
\section{Empirical Evaluation}
We empirically evaluate the performances of FDSP and GDP which are both applied to Max-Sum on four configurations of $n$-ary random DCOPs. Since BnB-MS and BnB-FMS are not generic algorithms and G-FBP is inferior to GDP \cite{khan2018generic}, we do not include them for comparison. The complexity of a $n$-ary DCOP can be quantified by the number of function-nodes, the average/maximal arity and the domain size \cite{Kim2013Improved,khan2018generic}. In addition to these parameters, we also find the number of variable-nodes can affect the complexity. Intuitively, given the number of function-nodes and the average arity per function-node, the graph density is actually determined by the number of function-nodes. Therefore, we introduce a new parameter called variable tightness (denoted as var\_T ) to depict the complexity from another perspective, which is defined as follows.
$$
\text{var\_T}= 1- \frac{\text{number of variable-nodes}}{\text{total number of arities}} 
$$
It can be concluded that given the function-node number and total arity number, the number of variable-nodes decreases as var\_T increases, which will generate a denser and more complex problem since each variable-node has to connect more function-nodes.

For each configuration other than the first one, we generate sparse factor graphs and dense factor graphs by randomly selecting var\_T from [0.1, 0.5] and (0.5, 0.9], respectively. In the first configuration, we set the number of function-nodes to 100 and the minimal arity to 2, and uniformly select the costs, the domain size and the maximal arity from [1, 100], [2, 10] and [2, 7], respectively. And, var\_T varies from 0.1 to 0.9. In the second one, we vary the maximal arity from 2 to 7. In the third configuration, we vary the number of function-nodes from 10 to 100. In the last one, we set the number of function-nodes to 50 and vary the domain size from 2 to 7. Also, we benchmark Max-Sum\_ADVP+FDSP to demonstrate the generalization of FDSP. To guarantee Max-Sum\_ADVP to converge, we alternate its directions every 100 iterations. All the omitted parameters except var\_T in each configuration are the same as the ones in the first configuration. For each of the setting, we generate 25 random instances and the results are averaged over all instances. The algorithms terminate after 200 iterations for each instance.
\begin{figure}
	\centering
	\includegraphics[scale=0.42]{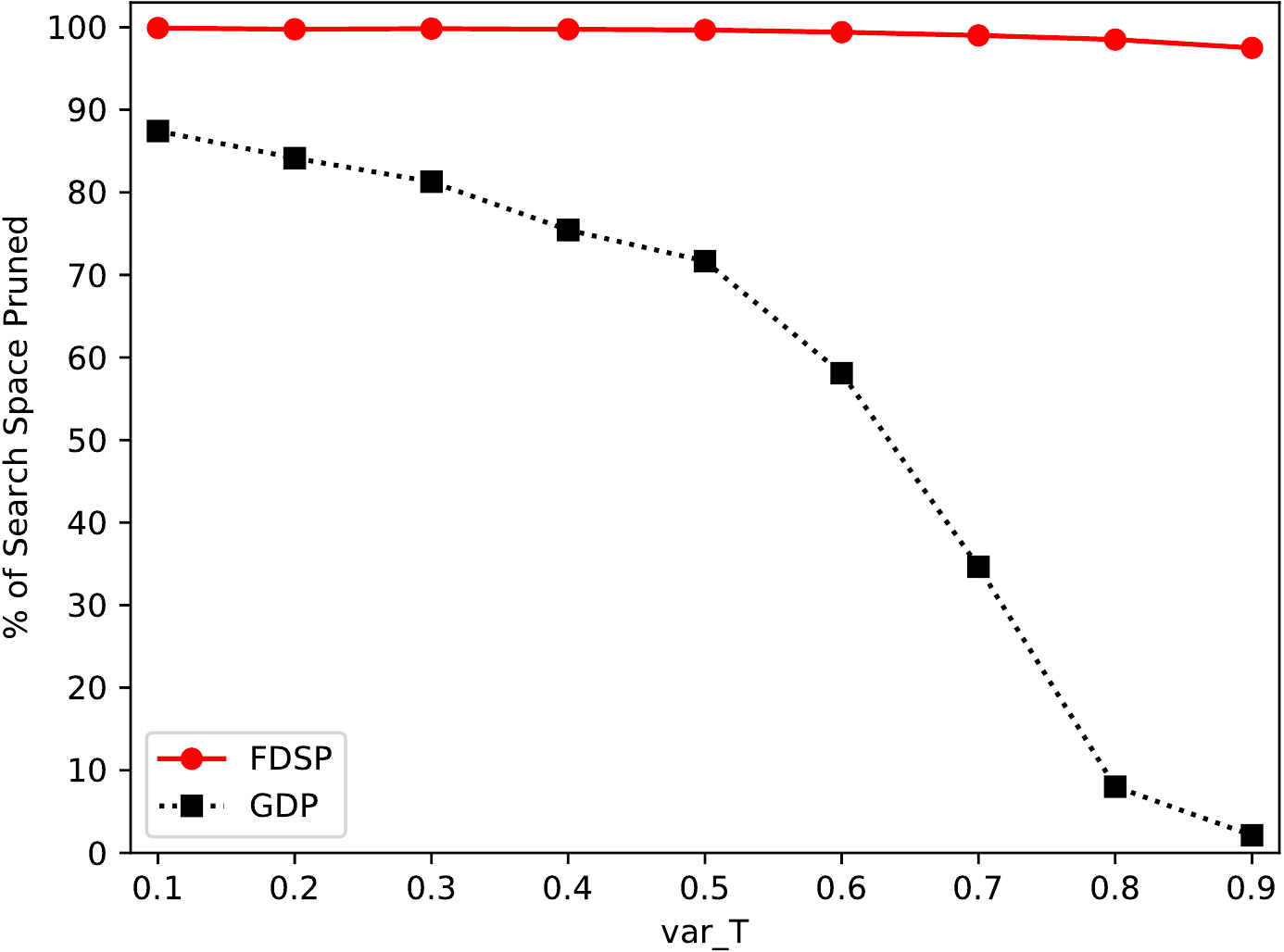}
	\caption{Performance comparison on different var\_T}
\end{figure}

Fig. 6 gives the comparison under different var\_T. It can be observed that FDSP clearly outperforms GDP under different var\_T, and the gap is widen as var\_T grows. Concretely, FDSP can prune at least 97\% of the search space while GDP only prunes at most 87\% of the search space when computing Eq. (2). Moreover, FDSP performs similarly as var\_T grows, which indicates FDSP is less sensitive to the complexity of problems. On the other hand, the performance of GDP decreases as var\_T increases, and GDP performs poorly when solving the  problems with var\_T$>=0.5$. This is because the sum of the difference between the maximal value and the value corresponding to the maximal local utility in each message will increase when the graph density increases as var\_T grows. As a result, GDP provides a large pruned range so as to prune only a small proportion of the search space.                              
\begin{figure}
	\centering
	\includegraphics[scale=0.42]{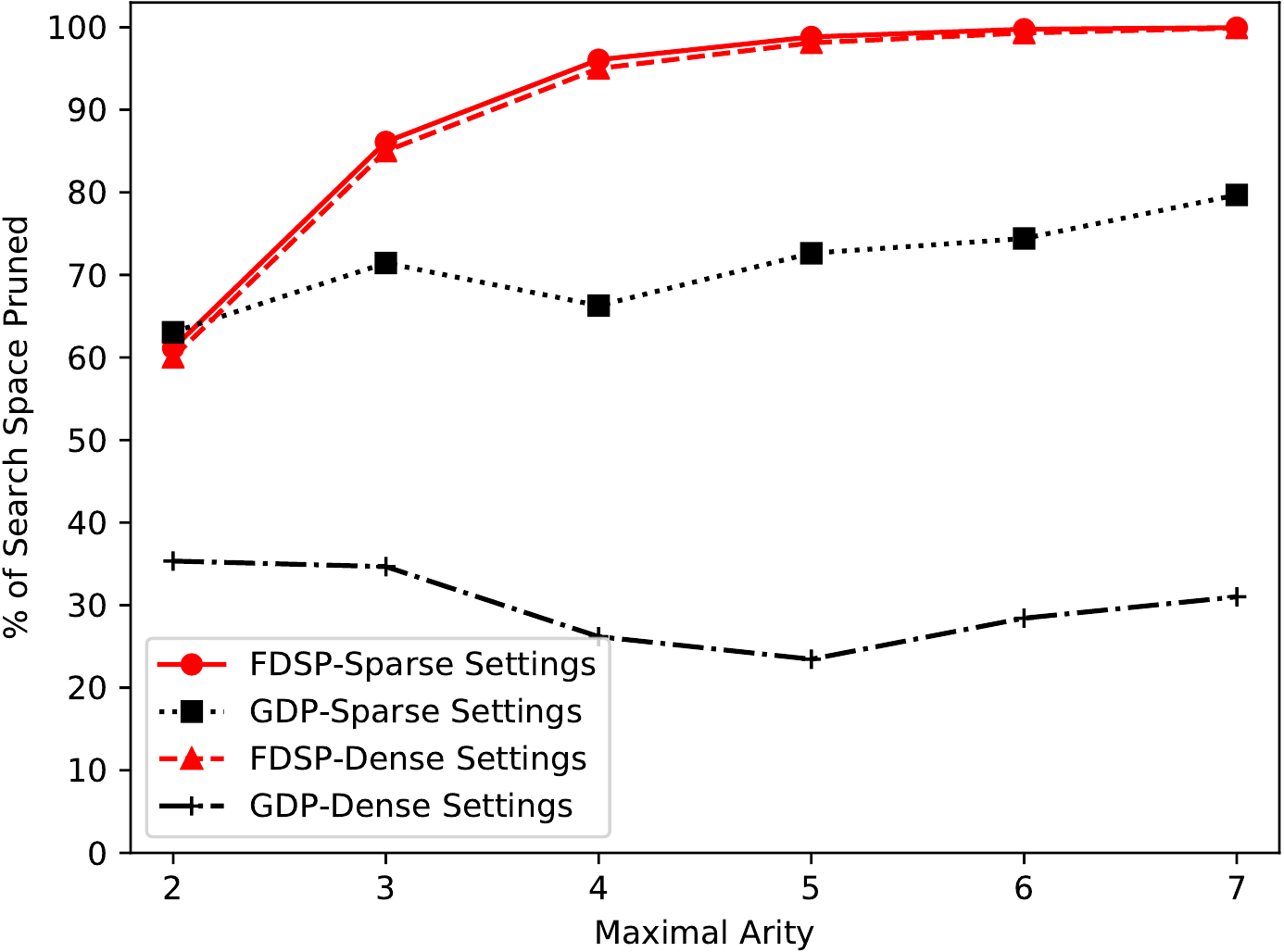}
	\caption{Performance comparison on different arities}
\end{figure}

Fig. 7 shows the performance comparison on different maximal arities. It can be concluded that FDSP outperforms GDP in both sparse and dense factor graphs, especially in dense factor graphs. FDSP prunes around 60\%-99\% of the search space in both sparse and dense factor graphs, while GDP can only prune at most 80\% and 36\%, respectively. That is because FDSP provides tighter bounds to make Max-Sum explore fewer combinations.   
\begin{figure}
	\centering
	\includegraphics[scale=0.42]{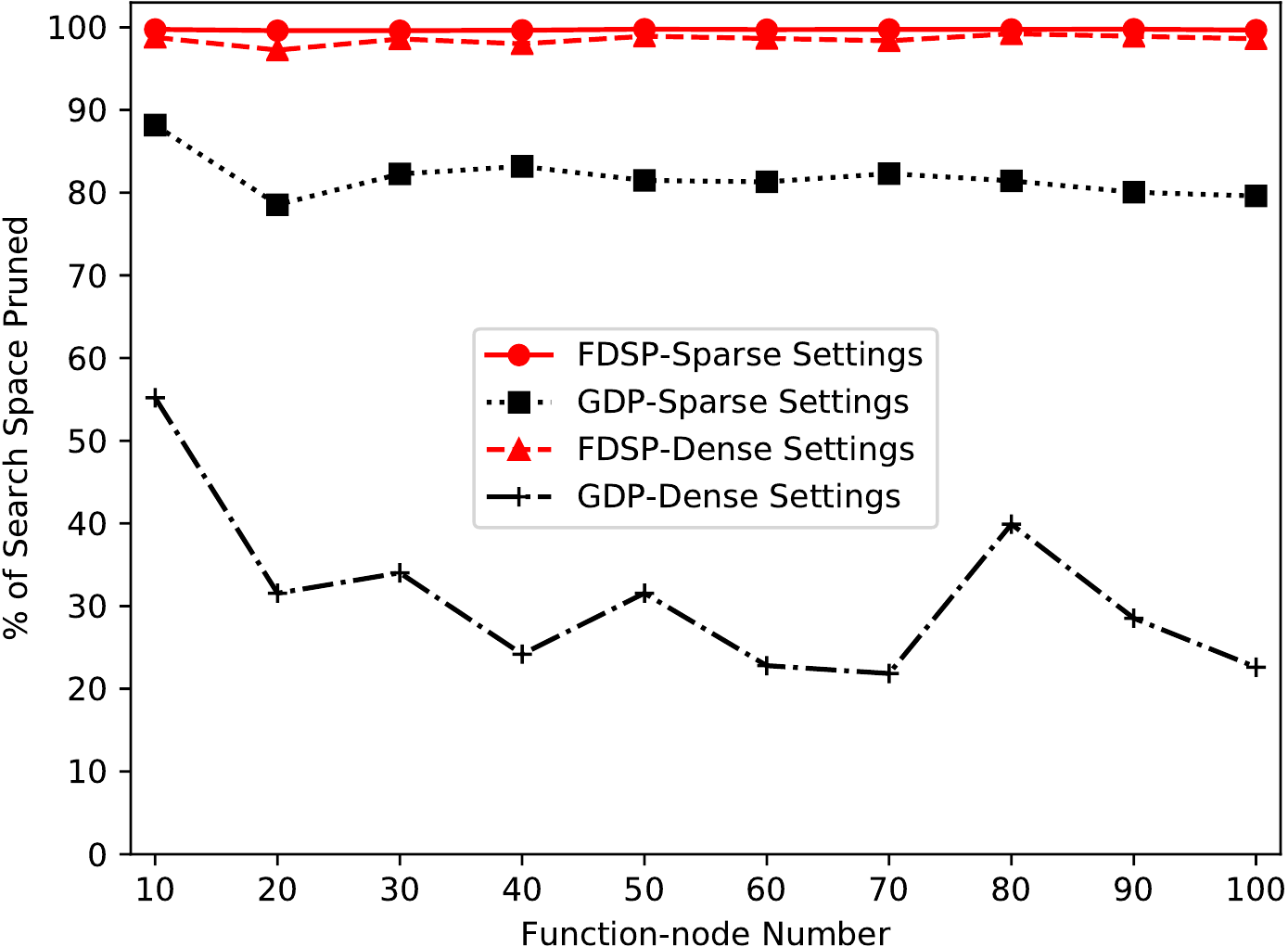}
	\caption{Performance comparison on different function-node numbers}
\end{figure}

Fig. 8 presents the results under different number of function-nodes. Similar to the first configuration, FDSP prunes at least 97\% of the search space in both sparse and dense factor graphs, while GDP can only prune at most 88\% and 55\% of the search space in sparse and dense factor graphs, respectively. This is because GDP is an one-shot pruning procedure and cannot use the learned experience from the assignment combinations explored to dynamically prune the search space.      
\begin{figure}
	\centering
	\includegraphics[scale=0.42]{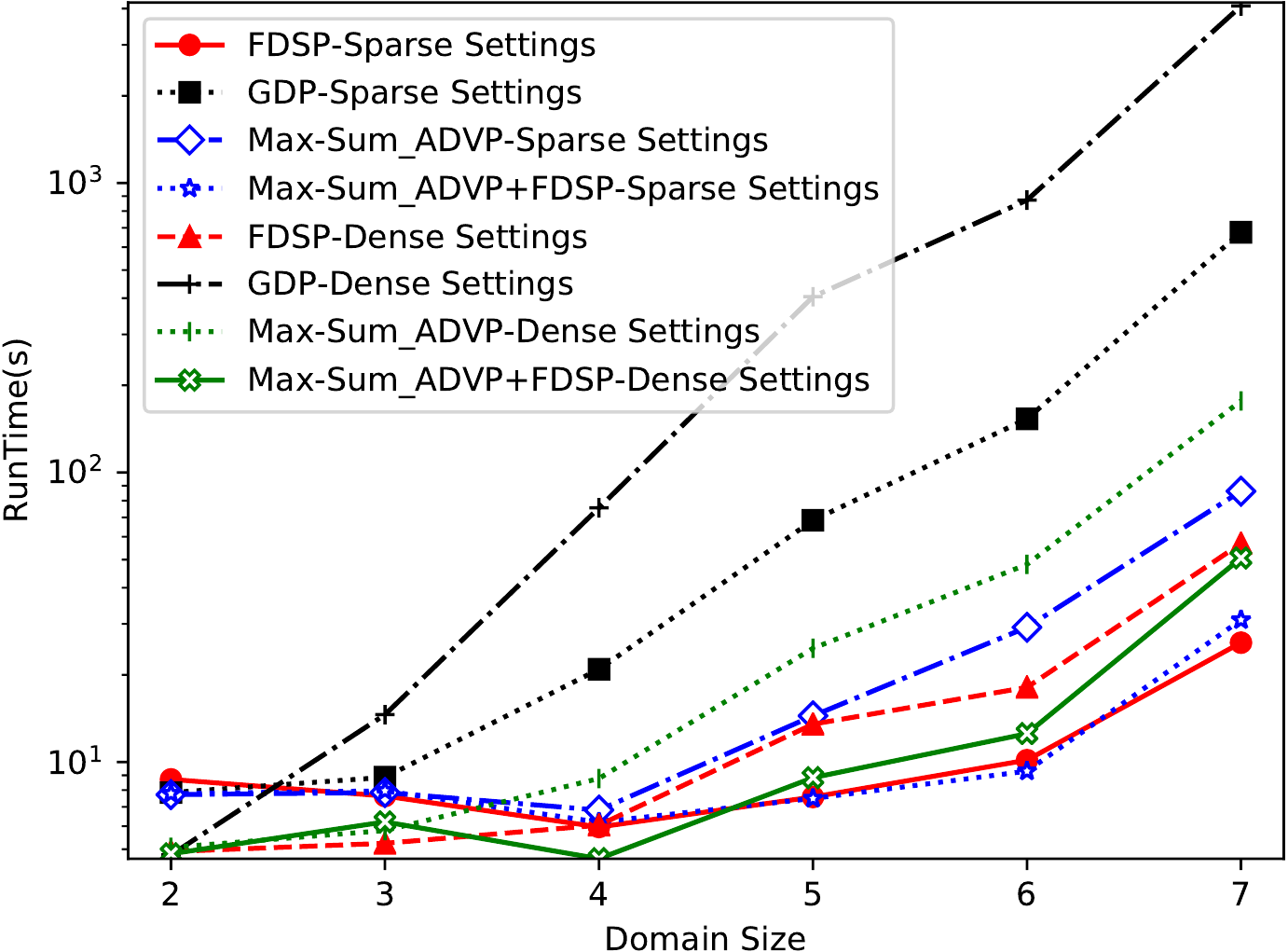}
	\caption{Runtime on different domain sizes}
\end{figure} 

Fig. 9 gives the runtime under different domain sizes. It can be seen that our FDSP exhibits great superiority over GDP and Max-Sum\_ADVP when solving the problems with large domain sizes, which indicates that FDSP can scale up well and only requires few computation efforts. GDP would perform even worse in practice since the runtime presented in Fig. 9 actually does not take sorting, which is quite expensive when the domain size is large, into consideration. Besides, one can easily observe that Max-Sum\_ADVP+FDSP is superior to Max-Sum\_ADVP when solving the problems with large domain sizes in sparse and dense factor graphs, which indicates that FDSP can also effectively accelerate the variants of Max-Sum.       
\section{Conclusion}
In this paper, we propose FDSP, a generic, fast and easy-to-use method based branch and bound, which significantly accelerates belief propagation based incomplete DCOP algorithms. Specifically, we first propose function decomposing (FD) to effectively compute the function estimation, which dramatically reduces the overheads in computing an upper bound of a partial assignment. Then, we further present state pruning (SP) based on branch and bound to reduce the search space. Besides, we theoretically prove that our bounds are monotonically non-increasing during the search process and FDSP never prunes the assignment with the maximum utility. Our experimental results clearly show that FDSP can prune around 97\%-99\% of the search space and only requires little time, especially for the large and complex problems.  
\section{Acknowledgment}
This research is funded by Chongqing Research Program of Basic Research and Frontier Technology (No. cstc2017jcyjAX0030), Fundamental Research Funds for the Central Universities (No. 2018CDXYJSJ0026) and Graduate Research and Innovation Foundation of Chongqing, China (Grant No. CYS18047).
\bibliography{wtfref}
\bibliographystyle{aaai}

\end{document}